\documentclass[10.5pt, a4paper, english,oneside]{article}
\usepackage[english]{babel}
\usepackage{amsmath, amssymb, multicol, amsfonts, amsthm, amscd, epsfig, enumerate}
\usepackage[latin1]{inputenc}
\usepackage[T1]{fontenc}
\usepackage{graphicx, tikz}
\usepackage{geometry}
\usepackage[round]{natbib}
\usepackage{hyperref}
\hypersetup{
    colorlinks,
    linkcolor={red!50!black},
    citecolor={green!50!black},
    urlcolor={blue!80!black}
}
\usepackage{epstopdf}
\usepackage{amscd}
\usepackage{multirow}
\usepackage{bm}
\usepackage{latexsym}
\usepackage{mathrsfs}
\usepackage{amsthm}
\usepackage{mathrsfs}
\usepackage{booktabs}
\usepackage{caption}
\usepackage{subfig}
\usepackage{url}
\usepackage{rotating}
\usepackage{multirow}
\usepackage{lscape}
\usepackage{longtable}
\usepackage{color}

\newtheorem{theorem}{Theorem}

\newtheorem{definition}[theorem]{Definition}

\newtheorem{proposition}[theorem]{Proposition}
\newtheorem{remark}[theorem]{Remark}

\definecolor{darkross}{rgb}{0.008,0.412,0.471}
\definecolor{middleross}{rgb}{0.012,0.580,0.663}
\definecolor{lightross}{rgb}{0.016,0.749,0.855}
\definecolor{darkblue}{rgb}{0.067,0.008,0.471}
\definecolor{middleblue}{rgb}{0.094,0.012,0.663}
\definecolor{lightblue}{rgb}{0.122,0.016,0.855}
\definecolor{darkpurple}{rgb}{0.471,0.008,0.412}
\definecolor{middlepurple}{rgb}{0.663,0.012,0.580}
\definecolor{lightpurple}{rgb}{0.855,0.016,0.749}
\definecolor{darkbrown}{rgb}{0.471,0.067,0.008}
\definecolor{middlebrown}{rgb}{0.663,0.094,0.012}
\definecolor{lightbrown}{rgb}{0.855,0.122,0.016}
\definecolor{darkolive}{rgb}{0.412,0.471,0.008}
\definecolor{middleolive}{rgb}{0.580,0.663,0.012}
\definecolor{lightolive}{rgb}{0.749,0.855,0.016}
\definecolor{darkgreen}{rgb}{0.008,0.417,0.067}
\definecolor{middlegreen}{rgb}{0.012,0.663,0.094}
\definecolor{lightgreen}{rgb}{0.016,0.855,0.122}
\definecolor{darkocre}{rgb}{0.471,0.298,0.008}
\definecolor{middleocre}{rgb}{0.663,0.420,0.012}
\definecolor{lightocre}{rgb}{0.855,0.541,0.016}

\def\qmo{``}
\def\qmc{''}
\def\qmcsp{'' }

\providecommand{\keywords}[1]{\textbf{Keywords:} #1}

\begin{document}

\title{Allocation of risk capital in a cost cooperative game induced by a modified Expected Shortfall}
\author{Mauro Bernardi\thanks{Department of Statistical Sciences,
University of Padua,
Via Cesare Battisti 241/243,
I-35121, Padua, Italy.}, \ Roy Cerqueti\thanks{Department of Economics and Law, University of Macerata, via Crescimbeni 20, I-62100, Macerata, Italy.}, \ Arsen Palestini\thanks{MEMOTEF, \ Sapienza University of Rome, Via del Castro Laurenziano 9, \ I-00161 Rome, \ Italy} }
\date{\today}

\maketitle

\begin{abstract}
%
\noindent The standard theory of coherent risk measures fails to
consider individual institutions as part of a system which might
itself experience instability and spread new sources of risk to the
market participants. In compliance with an approach adopted by
\cite{shapley_etal.1969}, this paper proposes a
cooperative market game where agents and institutions play the same
role can be developed. We take into account a multiple institutions
framework where some of them jointly experience distress events in
order to evaluate their individual and collective impact on the
remaining institutions in the market. To carry out this analysis, we
define a new risk measure ($SCoES$), generalising the Expected
Shortfall of \cite{acerbi.2002} and we characterise the riskiness
profile as the outcome of a cost cooperative game played by
institutions in distress (a similar approach was adopted by 
\citealt{denault.2001}). Each institution's marginal contribution to the
spread of riskiness towards the safe institutions in then evaluated
by calculating suitable solution concepts of the game such as the
Banzhaf--Coleman and the Shapley--Shubik values.
\end{abstract}

\par \bigskip


%
\section{Introduction}
\label{sec:intro}
%
The assessment of financial risk in a multi--institution framework
when some institutions are subject to systemic or non--systemic
distress is one of the main topics of the latest years which
received large attention from scholars in Mathematical Finance,
Statistics, Management, see, e.g., \cite{adrian_brunnermeier.2011,adrian_brunnermeier.2016}, \cite{billio_etal.2012}, \cite{acharya_etal.2012}, \cite{girardi_ergun.2013}, \cite{hautsch_etal.2014}, \cite{engle_etal.2014}, \cite{lucas_etal.2014}, \cite{bernardi_catania.2015}, \cite{sordo_etal.2015}, just to quote a few of the most relevant approaches. For an extensive and up to date survey on systemic risk measures, see \cite{bisias_etal.2012}, while the recent literature on systemic risk is reviewed by \cite{benoit_etal.2016}. Especially since 2008, Lehmann Brothers' collapse and the subsequent debt crisis contagion across Europe
raised crucial issues and questions to be addressed by new
macroeconomic and financial models. In particular, such events and
the related consequences provoked a wide increase in investigation
of instruments possibly suitable for risk evaluation, thereby
targeting the occurrence of contagion among institutions in distress
(see, among others, \citealt{drehmann_tarashev.2013} on the
theoretical side and \citealt{huang_etal.2012} and \citealt{huang_etal.2012b} on the
empirical side).\newline
\indent In this paper we are going to investigate circumstances where one or
more market institutions undergo financial distress and may spread
the contagion to the remaining safe institutions in the market.
Financial distress has a straightforward meaning: a condition of
dramatic financial instability of an institution, or lack of ability
to fully repay its creditors, or even the preliminary to bankruptcy
or liquidation. Although \qmo financial distress\qmcsp and \qmo default\qmcsp are
sometimes used interchangeably, there may be some differences, often
depending either on each country's specific bankruptcy law or on the
institution's rating and so on (see \citealt{altman_etal.2010} for an extensive discussion).\newline
\indent When a financial distress takes place in the market, each distressed
institution affects the safe ones and individually contributes to
the systemic risk. A crucial aspect of the contagion frameworks is
the assessment of each agent's (i.e. each institution's) marginal
effect on systemic risk. In order to evaluate such contributions, a
natural and intuitive tool is the Shapley value (see 
\citealt{shapley.1988}), as well as the Banzhaf value (see
\citealt{Banzhaf.1965} and \citealt{owen.1995}). It is interesting to note that,
although his world famous theories did not concern financial markets
specifically, L.S. Shapley, Nobel Laureate in Economic Sciences
in 2012, proposed a relevant application to market analysis in a paper with M. Shubik in 1969 (see \citealt{shapley_etal.1969}). In that paper, Shapley and
Shubik defined and investigated the \textit{direct markets} (see
\cite{shapley_etal.1969}, pp. 15--18), i.e., markets where, since each trader
possessed one initial unit of a personal commodity, traders might be
identified with commodities. As the Authors put it literally,
\qmo\textit{A canonical market form -- the direct market -- is introduced,
in which the commodities are in effect the traders themselves, made
infinitely divisible...}\qmcsp (p. 11), hence that was the first attempt
to model a cooperative game where agents had a twice nature: traders
and assets at the same time.\newline
\indent Hence, classical solution concepts widely used in Cooperative Game
Theory such as the Shapley value  and the Banzhaf value fit our
evaluation purposes. In more recent times, cooperative game theory
has already been employed for risk capital allocation by
\cite{denault.2001} and in the latest years by 
\cite{csoka_etal.2009} and \cite{abbasi_hosseinifard.2013}\footnote{A
mechanism for allocation of risk capital to subportfolios of pooled
liabilities has been proposed by Tsanakas and Barnett in 2003 (see
\citealt{tsanakas_barnett.2003}) and by Tsanakas in 2009 (see \citealt{tsanakas.2009}),
based on the Aumann--Shapley value.}. However, in none of the above
papers a characteristic value expressing the difference between an
expected value in a non--distressed case and an expected value in a
distressed--case has been taken into account.\newline
\indent The financial setup we are going to investigate involves a set of
possibly distressed institutions (for example a set of banks) and a
set of safe institutions. Our purpose is an assessment of the
contagion risk, i.e., the extent to which the safe institutions can
be affected and damaged by the financial distress of the unsafe
ones. For this purpose, risk measures have been
extensively theorized and analyzed, giving rise to a rich strand of
literature. Some of them might provide helpful techniques for risk
allocation among a number of distinct systemic agents. In this
respect, the coherent risk measures have been and are a crucial tool
for the assignment of risk capital, both for the richness of
properties due to their axiomatic structure and for their
adaptability to fit several financial frameworks. The basic studies
on coherent risk measures have been published by \cite{artzner_etal.1999} in 1999, who first provided an axiomatic characterization. Subsequently, in 2001 \cite{denault.2001}
developed the analysis by taking into account a cooperative game
structure, where players are the firms involved in risky financial
activities, and the risk induced by the net worth of firms is the
value of the characteristic function of the game, and the
subadditivity axiom naturally induces a cost cooperative game.\newline
\indent The coherence  of the Expected Shortfall (ES, hereafter) was taken into account by
\cite{acerbi_tasche.2002} in 2002, whereas an axiomatic approach was
adopted in 2005 by \cite{kalkbrener.2005}, to characterize the
related capital allocation procedures in portfolio management or
performance measurement. In recent years, further mathematical
results were achieved by \cite{stoica.2006} (equivalence between
properties of some coherent measures with no arbitrage conditions),
\cite{csoka_etal.2007, csoka_etal.2009} (refinement of measures
in a general equilibrium setup, stability of risk allocations),
\cite{kountzakis_polyrakis.2013} (applications to general
equilibrium models). Other relevant results in risk allocation
theory were provided in \cite{drehmann_tarashev.2013} and 
\cite{csoka_etal.2016}. Furthermore, as a slightly different tool, distortion risk measures
have been applied to a financial contagion framework 
by \cite{cherubini_mulinacci.2014}.\newline
\indent The new approach we are going to propose builds on an alternative
risk measure which is an extension of ES and which is supposed to
incorporate the effect of the distressed institutions on a safe one.
Such a measure, denoted by
$SCoES_{i\vert\mathcal{D}}^{\tau_1\vert \tau_2}$, represents the
expected value of a non--distressed institution $i$ conditional upon a
set $\mathcal{D}$ of institutions under distress, given two
different thresholds $\tau_1,\tau_2\in\left(0,1\right)$, specifically related to the institutions'
returns. Subsequently, in order to evaluate the marginal contribution of each
distressed institution to overall systemic risk, we will construct a
cost function. Such a function aims to measure the cost of risk by
evaluating the mean of the differences between the standard Expected
Shortfalls and the values of $SCoES$ of all non--distressed
institutions. This formula will be the characteristic value of a
cooperative game played by distressed institutions. Adopting a
classical approach, we will examine the properties of the game and
calculate the marginal impacts of each distressed institution,
pointing out the differences between the Shapley value and the
Banzhaf value of the game.\newline
\indent The introduced risk measurement framework is then adopted to empirically analyze the evolution of the $SCoES$ risk measure and the output of the associated cooperative game of risk for eight Eurozone sovereign Credit Default Swaps (CDS) over the period 2008--2014. Our example assumes the Germany as the \qmo safe\qmcsp country, and explores how the evolution of the remaining countries' debt conditions affect the health and financial stability of Germany considered as a \qmo proxy\qmcsp for the overall European system. As a byproduct of the proposed risk framework, the evolution of the total impact of the failure of the European system is monitored. Our results show that the risk of failure of the European system displays a transitory high level during period in between the Greece and Portugal bailouts (November 2010 -- May 2011), but effectively remains at high levels when the ECB president Mario Draghi announces the implementation of the Outright Monetary Transactions (OMT) and the European Stability Mechanism (ESM) in the months thereafter. It is only after the implementation of the OMT, and the ESM, that the systemic risk of Germany settles more permanently at a level that is roughly 60\% lower than during the crisis.\newline
%
%
\indent The remainder of the paper is structured as follows. In Section \ref{sec:risk_measures} the basic concepts about risk measures will be recalled, together with some axiomatic
details of coherent measures. Section \ref{sec:scoes_measure} introduces the indicators
that will be crucial to our analysis, whereas in Section \ref{sec:scoes_game} the
cooperative game of risk and its solution concepts are exposed and
discussed. Section \ref{sec:application} is devoted to the empirical application and Section \ref{sec:conclu}, conclusions and possible future developments are laid out.
%
%
%
\section{Risk measures in portfolio management}
\label{sec:risk_measures}
%
Here we will outline the standard notation for risk measurement in recent literature, largely borrowed from the seminal paper by Artzner \cite{artzner_etal.1999} and \cite{kalkbrener.2005}.
\par
Consider a finite set of states of nature $\Omega$, whose cardinality is $|\Omega|=m$. Call $X(\omega_j)$ a random variable indicating the final net worth of a position in state $\omega_j \in \Omega$ after a certain time interval, i.e., the possible profit and loss realization of a portfolio in state $\omega_j$. We can identify the set of all real-valued functions (which can also be viewed as the set of all risks) on $\Omega$ as $\mathbb{R}^m$, whose elements are of the kind $X=(X(\omega_1), \ldots, X(\omega_m))$. In the rigorous construction proposed in \cite{artzner_etal.1999}, a \textit{measure of risk} is a mapping $\rho: \mathbb{R}^m \longrightarrow \mathbb{R}$ such that $\rho(X)$ corresponds to the minimum amount of extra cash an agent has to add to her risky portfolio, in order to ensure that this investment is still acceptable to the regulator\footnote{This property relies on the \textit{acceptance sets}, which are axiomatized in \cite{artzner_etal.1999}.} when using a suitable reference instrument. The basic requirement concerns the price of the asset: 1 is the initial price of the asset and $r>0$ is the total return on the reference instrument at a final date $T$, in all possible states of nature.
In \cite{artzner_etal.1999} the Definition of a coherent measure is then provided based on an axiomatic structure.

\begin{definition}\label{coherentmeasure}
A function $\rho: \mathbb{R}^S \longrightarrow \mathbb{R}$ is called a \textbf{coherent measure of risk} if it satisfies the following axioms:
\begin{itemize}
\item[-] Monotonicity (M): for all $X, Y \in \mathbb{R}^m$ such that $Y \geq X$ (i.e. $Y(\omega_j) \geq X(\omega_j)$ for almost all $\omega_j \in \Omega$), $\rho(Y) \leq \rho(X)$.
\item[-] Subadditivity (S): for all $X, Y \in \mathbb{R}^m$, $\rho(X+Y) \leq \rho(X)+\rho(Y)$.
\item[-] Positive Homogeneity (PH): for all $X \in \mathbb{R}^m$ and $\lambda \in \mathbb{R}_+$, $\rho(\lambda X)=\lambda \rho(X)$.
\item[-] Translation Invariance (TI)\footnote{Property \textit{TI} is sometimes denoted as \textit{Risk Free Condition (RFC).}}: for all $X \in \mathbb{R}^m$ and $h \in \mathbb{R}$,  $$\rho(X+h r)=\rho(X)-h.$$
\end{itemize}
\end{definition}
\noindent We have to point out that the Value-at-Risk ($VaR$) satisfies all of them except Subadditivity, as is precisely exposed in \cite{artzner_etal.1999} as well as the standard Definition of $VaR$, which is recalled here:
\begin{definition}
Given $\tau \in [0,1]$ and the return on a reference instrument $r>0$, the $VaR$ at level $\tau$ of the final net worth $X$ with probability $\mathbb{P}\left\{\cdot \right\}$ is the negative of the quantile at level $\tau$ of $X/r$, i.e.
$$VaR_{\tau}=-\inf \left\{x \ | \ \mathbb{P} \left\{X \leq r \cdot x \right\} \geq \tau \right\}.$$
\end{definition}
\noindent Without loss of generality, in the remainder of the paper, $r$ will be normalized to $1$. \par
Back to coherent measures, when the outcomes are equiprobable, i.e. when the state of nature $\omega_j$ occurs with probability $p_j$ and $p_1= \cdots = p_m=1/m$, a special and relevant case can be treated, as was investigated by \cite{acerbi.2002}. In particular, we take into account an ordered statistics given by the ordered values of the $m$-tuple $X(\omega_1), \ldots, X(\omega_m)$, i.e. $\left\{\widetilde{X}_{1}, \ldots , \widetilde{X}_{m}\right\}$, rearranged in increasing order: $\widetilde{X}_{1} \leq \cdots \leq \widetilde{X}_{m}$. The definition of spectral measure of risk we are going to present is due to \cite{csoka_etal.2007}, who slightly modified Acerbi's original definition by employing a positive discount factor $\delta$, not necessarily equal to 1 as in \cite{acerbi.2002}.

 \begin{definition}\label{spectralriskmeasure}
 If the outcomes are equiprobable, given a vector $\Phi \in \mathbb{R}^m$ and a discount factor $\delta >0$, the risk measure $M_{\Phi}: \mathbb{R}^m \longrightarrow \mathbb{R}$ defined as follows:
 \begin{equation}
 M_{\Phi}(X)=-\delta \sum_{j=1}^m \Phi_j \widetilde{X}_{j},
 \label{MPhi}
 \end{equation}
 is called a \textbf{spectral measure of risk} if $\Phi$ satisfies the following axioms:
 \begin{itemize}
 \item[-] Nonnegativity (N1): $\Phi_j \geq 0$ for all $j=1, \ldots, m$.
 \item[-] Normalization (N2): $\Phi_1+ \cdots + \Phi_m=1$.
 \item[-] Monotonicity (M): $\Phi_j$ is non-increasing, i.e. for all $u, v \in \left\{1, \ldots, m \right\}$, $u<v$ implies $\Phi_u \geq \Phi_v$.
 \end{itemize}
 \end{definition}
A well-known spectral measure of risk is the one indicating the discounted average of the worst $\tau$ outcomes, i.e. the \textbf{$\tau$-Expected Shortfall} of $X$. For all $\tau \in \left\{1, \ldots, m \right\}$, it is given by
\begin{equation}
ES_{\tau}(X)=-\dfrac{\delta}{\tau} \sum_{j=1}^{\tau}
\widetilde{X}_{j}. \label{tauexpectedshortfalldiscrete}
\end{equation}
When $X$ is a continuous random variable, its formula reads as follows:
\begin{equation}
ES_{1-\lambda}(X)=-\dfrac{\delta}{1-\lambda} \int_{\widetilde{X} \leq VaR(\widetilde{X})=\tau} \widetilde{x} f(\widetilde{x}) dx,
\label{tauexpectedshortfallcontinuous}
\end{equation}
where $f(\cdot)$ is the law of $\widetilde{X}$ and $1-\lambda$ is
the related confidence level. The proof of the coherence of
$ES_{1-\lambda}(\cdot)$ in the continuous random variable setup can
be found in \cite{acerbi_tasche.2002}. 
%
\section{The $SCoES$ as a risk measure} 
\label{sec:scoes_measure}
%
Here we propose a generalisation of the \cite{adrian_brunnermeier.2016}'s $CoVaR$ and $CoES$, namely System--$CoVaR$ ($SCoVaR$) and System--$CoES$ ($SCoES$). The proposed risk measures aim to capture interconnections among multiple connecting market participants which is particularly relevant during periods of financial market crisis, when several institutions may contemporaneously experience distress instances.
\par Let $\mathcal{P}=\left\{1,\ldots,p\right\}$ be a set of $p$ institutions, and assume that the conditioning event is the distress of a subset of $\mathcal{P}$. Call
$\mathcal{D}=\left\{j_1,\ldots,j_d\right\}\subset \mathcal{P}$ the set of $d$ institutions potentially under distress, whose cardinality is such that $0< d < p$, hence meaning that at least one institution is not under distress and that at least one is under distress. If we consider a set $S \subseteq \mathcal{D}$ of distressed institutions, such set represents a group of institutions picked among the ones that may be under distress. As is usual in typical Cooperative Game Theory literature, we will denominate any group $S$ as a \textbf{coalition}.
\par \bigskip \bigskip \bigskip
\begin{center}
{\normalsize \setlength{\unitlength}{0.1mm}
\begin{figure}[h]
\begin{tikzpicture}
\tikz\draw[line width=0.6mm] (0,0) ellipse (2.5cm and 2cm);
\end{tikzpicture}
\begin{tikzpicture}
\tikz\draw[line width=1mm] (3.1,-1) circle (1.5cm);
\end{tikzpicture}
\begin{tikzpicture}
\tikz\draw[line width=0.4mm] (5.3,-3.5) ellipse (1cm and 0.6cm);
\end{tikzpicture}
\put(-550,310){\Huge{$\mathcal{D}$}}
\put(-760,370){\Huge{$\mathcal{P}$}}
\put(-550,450){\huge{$S$}}
\put(-1050,70){\textbf{Fig. 1}: The subset $\mathcal{D} \subset \mathcal{P}$ contains all possible coalitions $S$ of institutions in distress.}
\end{figure}}
\end{center}
Given the confidence levels $\tau_1,\tau_2$, we are going to define the $SCoVaR_{i\vert\mathcal{S}}^{\tau_1\vert \tau_2}$ of institution $i\in\mathcal{P}$ for all institutions not belonging to $\mathcal{D}$ as follows. We are going to assume that at least one distressed institution has a negative return, i.e., that there exists at least one $X_j<0$, to ensure possible positivity of the Expected Shortfall as has been defined by (\ref{tauexpectedshortfalldiscrete}), when $\delta, \tau >0$.\par
Call $F_{i | S} \left(\cdot, -VaR_{\tau_2}\left( \sum_{j_k \in S} X_{j_k} \right)\right)$ the joint cumulative density function of institution $i$ conditional on the set of institutions $S$ being under distress. Denote with $F_{S}(\cdot)$ the distribution function of the group of institutions $S$, provided that the involved random variable coincides with the sum of all returns of the institutions in $S$.

\begin{definition}\label{SCoVaR}
Let $X=\left(X_1,\ldots,X_p\right)$ be a vector of $p$ institution returns with probability $\mathbb{P}\{\cdot\}$. Given  a set $S \subseteq \mathcal{D}$ of institutions in distress and $\tau_1, \tau_2 \in [0, \ 1]$, for all $i\in\mathcal{P} \setminus \mathcal{D}$, the $SCoVaR_{i\vert S}^{\tau_1\vert \tau_2}$ is the following value:
    \begin{equation}
    SCoVaR_{i\vert S}^{\tau_1\vert \tau_2}=-\inf \left\{l \in \mathbb{R} \ | \ F_{i | S}(l, m) \geq \tau_1, \ \sum_{j_k \in S} X_{j_k} \leq m \right\},
    \label{SCOVAR2}
    \end{equation}
    where $m:=-\inf \left\{s \in \mathbb{R} \ | \ \mathbb{P} \left\{\sum_{j_k \in S} X_{j_k}\leq s \right\} \geq \tau_2 \right\}$.
\end{definition}
\noindent An alternative Definition of $SCoVaR_{i\vert\mathcal{S}}^{\tau_1\vert \tau_2}$ is the following one.

\begin{definition}\label{SCoVaR2}
Let $X=\left(X_1,\ldots,X_p\right)$ be a vector of $p$ institution returns with probability $\mathbb{P}\{\cdot\}$. Given  a set $S \subseteq \mathcal{D}$ of institutions in distress and $\tau_1, \tau_2 \in [0, \ 1]$, for all $i\in\mathcal{P} \setminus \mathcal{D}$, the $SCoVaR_{i\vert S}^{\tau_1\vert \tau_2}$ is the
maximum value $X_i^*$ taken by $X_i$ such that
\begin{equation}
\dfrac{\mathbb{P}\left\{\left\{X_i \leq X_i^* \right\} \bigcap  \left\{\sum_{j_k \in S} X_{j_k}\leq -VaR_{\tau_2}\left( \sum_{j_k \in S} X_{j_k} \right)\right\}\right\}}{\mathbb{P} \left\{\sum_{j_k \in S} X_{j_k}\leq -\text{VaR}_{\tau_2}\left( \sum_{j_k \in S} X_{j_k} \right)\right\}}  \geq \tau_1.
\label{SCOVAR}
\end{equation}
\end{definition}
\noindent Basically, $SCoVaR_{i\vert S}^{\tau_1\vert \tau_2}$ is the Value-at-Risk of an institution subject to the condition that the sum of the realizations of the institutions under distress do not exceed the Value-at-Risk of their sum, when two different confidence levels are in general taken into account. The two following Remarks aim to point out two circumstances where $SCoVaR$ coincides with the marginal $VaR$.

\begin{remark}\label{somma}
In this context it is quite natural to consider the sum as aggregated measure of risk since we are considering profits and losses. Of course alternative definitions are possible, as for example, the maximum loss of the distressed institutions, see, e.g., \cite{bernardi_etal.2016b} and discussion therein. 
\end{remark}

\begin{remark} \label{indipendence}
If all the returns of the institutions in $W$ are independent of all the returns of institutions in $\mathcal{D}$, then the joint c.d.f. $F_{i | S}(\cdot)$ becomes $F_i(\cdot)$, i.e. the c.d.f. of institution $i \in W$. Consequently,
\begin{equation*}
SCoVaR_{i\vert S}^{\tau_1\vert \tau_2}=-\inf \left\{l \in \mathbb{R} \ | \ \mathbb{P}\left\{X_i \leq l \right\} \geq \tau_1\right\}=VaR_{\tau_1}(X_i).
\end{equation*}
\end{remark}

\begin{remark} \label{SCoVaRemptyset}
When no institution is under distress, $S=\emptyset$, i.e. $X_{j_k}=0$ for all $j_k \in S$. In this case, (\ref{SCOVAR}) is well-defined too and in particular it collapses to the standard VaR. Namely, $$\mathbb{P} \left\{\sum_{j_k \in S} X_{j_k}\leq -VaR_{\tau_2}\left( \sum_{j_k \in S} X_{j_k} \right)\right\}=1 \ \Longrightarrow \ SCoVaR_{i\vert \emptyset}^{\tau_1\vert \tau_2}=VaR_{\tau_1}(X_i).$$
\end{remark}

\noindent The $SCoVaR$ is particularly useful to formulate the risk measure we are going to investigate.

\begin{definition}
Given $\tau_1, \tau_2 \in [0, \ 1]$ and a set $S$ of institutions in distress, the $SCoES_{i\vert S}^{\tau_1\vert \tau_2}$ is the expected value of institution $i\in\mathcal{P}\setminus \mathcal{D}$,
provided that it does not exceed $SCoVaR_{i\vert S}^{\tau_1\vert \tau_2}$ and  conditional upon the set of institutions $S$ being at the level of their joint $ES_{\tau_2}$-level:
\begin{equation}
SCoES_{i\vert S}^{\tau_1\vert \tau_2}\equiv \mathbb{E} \left[X_i\ | \ X_i \leq SCoVaR_{i\vert S}^{\tau_1\vert \tau_2}, \ \sum_{j_k \in S} X_{j_k}\leq ES_{\tau_2}\left(\sum_{j_k \in S} X_{j_k}\right)\right].
\label{SCoES2}
\end{equation}
\end{definition}

\noindent As in Remark \ref{SCoVaRemptyset}, (\ref{SCoES2}) can be evaluated when no distress occurs too, collapsing to the standard Expected Shortfall:
$$SCoES_{i\vert \emptyset}^{\tau_1\vert \tau_2}=\mathbb{E} \left[X_i\ | X_i \leq VaR_{\tau_1}(X_i) \right]=ES_{\tau_1}(X_i).$$
A short explanation may be helpful to clarify the formulation of (\ref{SCoES2}): the condition $\sum_{j_k \in S} X_{j_k}\leq ES_{\tau_2}\left(\sum_{j_k \in S} X_{j_k}\right)$ is always verified when all $X_{j_k}<0$. On the other hand, just one negative institution return is enough to determine an open interval for $\delta$ such that the condition holds. More precisely, if $X_{j_1}, \ldots, X_{j_l}$ are negative, where $j_k \in S$ for $k=1, \ldots, l$, the condition boils down to: $\delta \geq - \tau_2 \frac{\sum_{j_k \in S} X_{j_k}}{\sum_{k=1}^l X_{j_k}}$.
Such an estimate is trivially true whenever all institution returns in a coalition $S$ are negative\footnote{Also note that the estimates on $\delta$ are as many as the possible coalitions $S$ except the empty set, i.e. $2^{|\mathcal{D}|}-1$, hence there are at most $2^{|\mathcal{D}|}-1$ levels of $\delta$ that must be exceeded. Because the choice of $\delta$ in the definition of (\ref{tauexpectedshortfalldiscrete}) is arbitrary, taking the maximum level among such values implies that such condition in (\ref{SCoES2}) is always satisfied, consequently $SCoES_{i\vert S}^{\tau_1\vert \tau_2}\equiv \mathbb{E} \left[X_i\ | \ X_i \leq SCoVaR_{i\vert S}^{\tau_1\vert \tau_2}\right]$ for all $i \in \mathcal{P}\setminus \mathcal{D}$.}. \par
Getting to analyze possible relations between $SCoVaR$ and $SCoES$, we can prove some results.

\begin{proposition}\label{propertyofSCoES}
Given two coalitions $S, S^\prime \in 2^{\mathcal{D}}$, if the following hypotheses are verified:
\begin{enumerate}
\item $SCoVaR_{i\vert S^\prime}^{\tau_1\vert \tau_2}>SCoVaR_{i\vert S}^{\tau_1\vert \tau_2}$;
\item $\sum_{j_k \in S^\prime} X_{j_k}\leq ES_{\tau_2}\left(\sum_{j_k \in S^\prime} X_{j_k}\right)$;
\item $\sum_{j_k \in S} X_{j_k}\leq ES_{\tau_2}\left(\sum_{j_k \in S} X_{j_k}\right)$,
\end{enumerate}
then
$SCoES_{i\vert S^\prime}^{\tau_1\vert \tau_2} \geq SCoES_{i\vert S}^{\tau_1\vert \tau_2}$.
\end{proposition}
\begin{proof}
If $S,S^\prime \in 2^{\mathcal{D}}$, the first hypothesis ensures that
$$\mathbb{E} \left[X_i\ | \ X_i \leq SCoVaR_{i\vert S^\prime}^{\tau_1\vert \tau_2}\right] \geq \mathbb{E} \left[X_i\ | \ X_i \leq SCoVaR_{i\vert S}^{\tau_1\vert \tau_2}\right],$$
whereas the second and the third hypotheses guarantee that
$SCoES_{i\vert S^\prime}^{\tau_1\vert \tau_2}$ and $SCoES_{i\vert S}^{\tau_1\vert \tau_2}$ are well-defined, consequently
$$SCoES_{i\vert S^\prime}^{\tau_1\vert \tau_2} \geq SCoES_{i\vert S}^{\tau_1\vert \tau_2}.$$
\end{proof}
\noindent In the remainder of the paper, whenever there is no misunderstanding, we are going to simply denote the above quantities with $ES$, $SCoES$, $VaR$ and $SCoVaR$, to lighten the notation.
%
\section{A cooperative game for risk allocation induced by $SCoES$}
\label{sec:scoes_game}
%
The risk allocation problems were introduced by 
\cite{denault.2001}, where the problem of allocating the risk of
a given firm, as measured by a coherent measure of risk, among its
$N$ constituents, was taken into account, closely resembling the
typical Cooperative Game Theory approach. In \cite{denault.2001}
cooperative games of cost were employed to model risk allocation
problems, and the chosen solution concepts were the Shapley value
and the Aumann--Shapley value. Such approach was subsequently adopted
and improved in \cite{csoka_etal.2009}, where risk allocation games and
totally balanced games are compared to ensure the existence of a
stable allocation of risk. In particular, they define a risk
environment characterized by a set of portfolios, a set of states of
nature, a discrete probability density of realization of states, a
matrix of realization vectors and a coherent measure of risk, from
which they construct and analyze a risk allocation game. In both
approaches, the portfolios of a firm are looked upon as the players
of a subadditive cooperative game.

In our setting, we are ready to apply the measures defined in
Section \ref{sec:scoes_measure} to institutional circumstances where distress
occurs, and in particular we are going to rely on some typical tools
borrowed from Cooperative Game Theory. In more details, we will look
upon any possible set of distressed institutions as a coalition of a
cooperative game (or TU-game, see \citealt{owen.1995}). The effect of
distress on the remaining institutions, corresponding to risk of
contagion, will be evaluated by means of a cost function.
\par Call $W$ the set of institutions not belonging to
$\mathcal{D}$, i.e. $W=\mathcal{P} \setminus \mathcal{D}$. We can
evaluate the cost of risk induced by any coalition $S \subseteq
\mathcal{D}$ by taking a weighted arithmetic mean over all
differences between the unconditioned $ES_{\tau_1}$ and the $SCoES$
for all institutions in $W$, i.e.
\begin{equation}
c_W(S)=\dfrac{\sum_{i \in W}\alpha_i \left[ES_{\tau_1}(X_i)-SCoES_{i\vert S}^{\tau_1\vert \tau_2}\right]}{|W|},
\label{cW}
\end{equation}
where $\alpha_i \geq 0$ for all $i=1, \ldots, |W|$,
$\sum_{i=1}^{|W|}\alpha_i=|W|$, for all $S \subseteq \mathcal{D}$.

In the frame of risk allocation, we introduce a cooperative game
$\Gamma=(c_W, \mathcal{D})$, where $\mathcal{D}$ represents the the
set of involved portfolios and $c:2^\mathcal{D} \longrightarrow
\mathbb{R}$ is as in (\ref{cW}), and assigns a cost to each
coalition $S \subseteq \mathcal{D}$.

The cooperative game approach appears very suitable, in that in an
uncertain financial framework it allows to take into account all
possible combinations of institutions undergoing distress.

Moreover, the couple $(c_W, \mathcal{D})$ actually defines a
cooperative game. In fact, when no institution in $\mathcal{D}$ is
under distress, then $c_W(\emptyset)=0$ because all the differences
in the numerator of (\ref{cW}) vanish. Essentially, this hypothesis,
which is necessary to define a cooperative game on $\mathcal{D}$,
may have a clear-cut financial interpretation: all the safe
institutions are collected in $W$, meaning that all of them are
secured beyond a reasonable doubt. They can be viewed as states or
companies issuing either government bonds or securities guaranteed
by top-quality collateral, namely all the kinds of agents which do
not involve any risk factors. Also note the positive sign in
(\ref{cW}): in standard payoff cooperative games such sign is
reversed. But because we are assuming that below the confidence
level $\tau_1$ some realizations $X_{J_k}$ are negative, positivity
of $ES_{\tau_1}(\cdot)$ is ensured, hence the level of risk induced
by distress can be positive. However, it may occur that for some
coalition $S$, $c_W(S)$ is non-positive, but this would mean that
the contagion is even less likely to spread from such a group of
institutions to the non-distressed ones.\newline
\indent Formula (\ref{cW}) needs some further explanation, in terms of what the
differences between $ES$ and $SCoES$ actually measure. Each
difference provides the spread between the standard risk and the
risk which is correlated to the distress of a coalition $S$,
composed of one institution at least. In order to completely assess
the risk effect caused by any coalition $S$, the sum of those
differences is taken over the whole set of safe institutions $W$.
Perhaps some structural differences may occur among the safe
institutions, including insurance contracts, implementation of
hedging strategies, and so on. Such heterogeneity can be captured by
weights $\alpha_i$ in (\ref{cW}), which can also be interpreted as
directly dependent of each single institution, in compliance with
its size or its systemic relevance. In order to lighten the
notation, we are going to hypothesize a simplified scenario where
all institutions' weights are equal, then we are going to posit
$\alpha_1=\cdots=\alpha_{|W|}=1$.\newline
\indent The issue concerning the properties of the game (\ref{cW}) is
somewhat complex, due to the fact that the $SCoES$ of an institution
subject to an external distress is a kind of measure of correlation,
or also a measure of how distress spreads its contagion towards
non-distressed institutions. Consequently, the standard axioms
associated to coherent risk measures can hardly be demonstrated.
Instead of an axiomatization, we are going to outline some
characteristics of $c_W(\cdot)$, which are listed in the next
Propositions. Some of these properties resemble the axioms stated by
\cite{denault.2001}.

\begin{proposition}\label{propertiesofcW1}
Given two coalitions $S, S^\prime \in 2^{\mathcal{D}}$, if the following hypotheses are verified:
\begin{enumerate}
\item $SCoVaR_{i\vert S^\prime}^{\tau_1\vert \tau_2}>SCoVaR_{i\vert S}^{\tau_1\vert \tau_2}$;
\item $\sum_{j_k \in S^\prime} X_{j_k}\leq ES_{\tau_2}\left(\sum_{j_k \in S^\prime} X_{j_k}\right)$;
\item $\sum_{j_k \in S} X_{j_k}\leq ES_{\tau_2}\left(\sum_{j_k \in S} X_{j_k}\right)$,
\end{enumerate}
then
$c_W(S^\prime) \leq c_W(S)$.
\end{proposition}

\begin{proof}
It follows directly from the inequality in Proposition \ref{propertyofSCoES}.
\end{proof}

In the following Proposition, the expression $\lambda S$ means that
all the institution returns in $S$ are multiplied by $\lambda$, i.e.
$\lambda X_{j_k}$, for all $j_k \in S\subseteq\mathcal{D}$.

\begin{proposition}\label{propertiesofcW2}
For each $\lambda \in \mathbb{R}_+$,
$c_W(\lambda S)=c_W(S)$.
 \end{proposition}

\begin{proof}
Applying Definition \ref{SCoVaR2} to $\lambda S$, for all $\lambda>0$, implies that
$SCoVaR_{i\vert \lambda S}^{\tau_1\vert \tau_2}$ is the maximum $X_i^*$ such that
$$\dfrac{\mathbb{P}\left\{\left\{X_i \leq X_i^* \right\} \bigcap  \left\{\sum_{j_k \in S} \lambda X_{j_k}\leq -VaR_{\tau_2}\left( \sum_{j_k \in S} \lambda X_{j_k} \right)\right\}\right\}}{\mathbb{P} \left\{\sum_{j_k \in S} \lambda X_{j_k}\leq -VaR_{\tau_2}\left( \sum_{j_k \in S} \lambda X_{j_k} \right)\right\}}  \geq \tau_1.$$
By the positive homogeneity of $VaR$ (see either \citealt{artzner_etal.1999} or \citealt{mcneil_etal.2015}, p. 74), that expression coincides with
(\ref{SCOVAR}), hence $SCoVaR_{i\vert \lambda S}^{\tau_1\vert
\tau_2}=SCoVaR_{i\vert  S}^{\tau_1\vert \tau_2}$.\newline 
\noindent For all $S\in 2^{\mathcal{D}}$, linearity of $ES$ can be employed in the
expression of $SCoES_{i\vert \lambda S}^{\tau_1\vert \tau_2}$:
\begin{align}
SCoES_{i\vert \lambda S}^{\tau_1\vert \tau_2}&\equiv \mathbb{E} \left[X_i\ | \ X_i \leq SCoVaR_{i\vert \lambda S}^{\tau_1\vert \tau_2}, \ \sum_{j_k \in S} \lambda X_{j_k}\leq ES_{\tau_2}\left(\sum_{j_k \in S} \lambda X_{j_k}\right)\right]\nonumber\\
&=\mathbb{E} \left[X_i\ | \ X_i \leq SCoVaR_{i\vert S}^{\tau_1\vert \tau_2}, \ \sum_{j_k \in S} X_{j_k}\leq ES_{\tau_2}\left(\sum_{j_k \in S} X_{j_k}\right)\right]\nonumber\\
&=SCoES_{i\vert S}^{\tau_1\vert \tau_2}.\nonumber
\end{align}
Finally, $c_W(\lambda S)$ can be written as follows:
\begin{align}
c_W(\lambda S)&=\dfrac{\sum_{i \in W} \left[ES_{\tau_1}(X_i)-SCoES_{i\vert \lambda S}^{\tau_1\vert \tau_2}\right]}{|W|}\nonumber\\
&=\dfrac{\sum_{i \in W} \left[ES_{\tau_1}(X_i)-SCoES_{i\vert S}^{\tau_1\vert \tau_2}\right]}{|W|}
=c_W(S).
\end{align}
\end{proof}

\begin{proposition}
 For all $S \subseteq \mathcal{D}$ such that all the returns $X_i$, where $i \in W$, are independent of all returns $X_{j_k}$, where $j_k \in S$, $c_W(S)=c_W(\emptyset)=0$.
\end{proposition}

\begin{proof}
For all institution returns $X_i$ which are independent of all $X_{j_k}$ (see Remark \ref{indipendence}), where $i \in W$ and $j_k \in S \subseteq \mathcal{D}$, we have that $SCoES_{i\vert S}^{\tau_1\vert \tau_2}=SCoES_{i\vert \emptyset}^{\tau_1\vert \tau_2}=ES_{\tau_1}[X_i]$. If this holds for all $i \in W$, the proof is complete.
\end{proof}

\noindent Subadditivity is a key feature of cost allocation games. Recall that
$\Gamma$ is \textit{subadditive} when its cost function is
subadditive, i.e. for all $S, T \in 2^\mathcal{D}$ such that $S \cap
T =\emptyset$, $c_W(S \cup T) \leq c_W(S)+c_W(T)$  (e.g. \citealt{anily_haviv.2014}). As is shown in the next Proposition, the
game$(c_W, \mathcal{D})$ is not always subadditive. In particular,
some assumptions on the related $ES$ and $SCoES$ are supposed to
hold.

\begin{proposition}\label{propertiesofcW3}
If $\forall \ i \in W$, $ES_{\tau_1}(X_i)>0$ and $SCoES$ is superadditive, i.e.
\begin{equation*}
SCoES_{i\vert S \cup T}^{\tau_1\vert \tau_2} \geq SCoES_{i\vert S}^{\tau_1\vert \tau_2} + SCoES_{i\vert T}^{\tau_1\vert \tau_2},
\end{equation*}
then $(c_W, \mathcal{D})$ is subadditive.
\end{proposition}

\begin{proof}
By definition of subadditivity, we can note that for all $S,T \in 2^{\mathcal{D}}$, where $S \cap T=\emptyset$:
\begin{align}
&c_W(S \cup T)-c_W(S)-c_W(T)=\nonumber\\
&\qquad\qquad=\dfrac{\sum_{i \in W}\left[SCoES_{i\vert S}^{\tau_1\vert \tau_2} + SCoES_{i\vert T}^{\tau_1\vert \tau_2} -SCoES_{i\vert S \cup T}^{\tau_1\vert \tau_2}-ES_{\tau_1}(X_i)\right]}{|W|},\nonumber
\end{align}
then, if $\forall \ i \in W$, $ES_{\tau_1}(X_i)>0$ and $SCoES$ is superadditive, i.e.
\begin{equation*}
SCoES_{i\vert S \cup T}^{\tau_1\vert \tau_2} \geq SCoES_{i\vert S}^{\tau_1\vert \tau_2} + SCoES_{i\vert T}^{\tau_1\vert \tau_2},
\end{equation*}
then $c_W(\cdot)$ is subadditive.
\end{proof}
\noindent However, it is difficult and restrictive to impose this condition, because the very definition of $ES_{\tau_1}$ allows both its possible positivity and negativity, depending on the level at which the worst outcomes are taken.
%
\subsection{Risk allocation}
\label{sec:risk_allocation}
%
A typical and well-known application of cooperative games is the determination of suitable allocations among players who get a share of a total amount, which is a benefit when they play a payoff game (generally superadditive) and a cost when they play a cost game (generally subadditive). In this case, by Proposition \ref{propertiesofcW3}, subadditivity is not ensured, but the characteristic function of the game represents a contagion risk induced by distress of some institutions, hence its interpretation as a cost game sounds intuitive and natural.
\par
A complete presentation of the several solution concepts and allocation rules in cooperative games can be found in \cite{owen.1995}. The two main values we are going to apply to our setup are the Shapley-Shubik (first introduced in 1953, see \citealt{shapley.1988}) and the Banzhaf-Coleman (which was formulated in 1965, see \citealt{Banzhaf.1965}) values. What follow are the expressions of such allocation principles when employing the characteristic function $c_W(\cdot)$.
\par
The \textbf{Shapley value} of the game $(c_W, \mathcal{D})$ is given by the $d$-dimensional vector $\Phi(c_W)=(\phi_1(c_W), \ldots, \phi_d(c_W))$ such that:
\begin{equation}
\phi_{j_k}(c_W)=\sum_{j_k \in S, \ S \subseteq \mathcal{D}} \dfrac{(d-|S|)!(|S|-1)!}{d!}\left[\dfrac{\sum_{i \in W}\left[SCoES_{i\vert S\setminus\left\{j_k\right\}}^{\tau_1\vert \tau_2}-SCoES_{i\vert S}^{\tau_1\vert \tau_2}\right]}{|W|} \right],
\label{ShapleycW}
\end{equation}
for all $j_k \in \mathcal{D}$.
\par
On the other hand, the \textbf{Banzhaf value} of $(c_W, \mathcal{D})$ is the $d$-dimensional vector $\beta(c_W)=(\beta_1(c_W), \ldots, \beta_d(c_W))$ such that:
\begin{equation}
\beta_{j_k}(c_W)=\dfrac{1}{2^{d-1}}\sum_{j_k \in S, \ S \subseteq \mathcal{D}} \left[\dfrac{\sum_{i \in W}\left[SCoES_{i\vert S\setminus\left\{j_k\right\}}^{\tau_1\vert \tau_2}-SCoES_{i\vert S}^{\tau_1\vert \tau_2}\right]}{|W|} \right],
\label{BanzhafcW}
\end{equation}
for all $j_k \in \mathcal{D}$.
\par
Their respective axiomatizations\footnote{There exists a large number of contributions on axiomatizations of values in literature, see for example \cite{feltkamp.1995} and \cite{van_den_brink_etal1998}.} point out a crucial difference between (\ref{ShapleycW}) and (\ref{BanzhafcW}): the Shapley value satisfies the \textit{efficiency} axiom\footnote{Nonetheless, some recent contributions have been published on the Shapley value without the efficiency axiom, see \cite{eini_haimanko.2011} for simple voting games and \cite{casajus.2014a} for different classes of games.}, i.e. $\sum_{k=1}^d \phi_{j_k}(c_W)=c_W(\mathcal{D})$, whereas the Banzahf value does not, except when $d=2$. On one hand, such axiom conveys the idea that there is an aggregate amount of risk capital to be apportioned among institutions in distress. On the other hand, perhaps it is helpful to avoid thinking of risk as a unique object to be divided, given its specific characteristics. Loosely speaking, we stress that both values can be employed based on good motivations.
\par
Clearly, we can say that $j_k$ is a \textbf{dummy institution} if and only if $\forall \ i \in W$, $\forall S \subseteq \mathcal{D}$, $SCoES_{i\vert S\setminus\left\{j_k\right\}}^{\tau_1\vert \tau_2}=SCoES_{i\vert S}^{\tau_1\vert \tau_2}$. The economic meaning of a dummy institution is simple: its marginal contribution to overall contagion is always zero.
\par
A special discussion should be devoted to the so-called \textbf{no undercut} property (see \cite{denault.2001}, Def. 3), which can be reformulated as follows: given an allocation $(K_{j_1}, \ldots, K_{j_d})$ for the game $(c_W, \mathcal{D})$, for all $S \subseteq \mathcal{D}$, the inequality
\begin{equation}
\sum_{j_k \in S} K_{j_k} \leq c_W(S),
\label{noundercut}
\end{equation}
must hold. The condition (\ref{noundercut}) has a twofold meaning. The first one is technical: any allocation $(K_{j_1}, \ldots, K_{j_d})$ satisfying it for all $S$ is in the core of the cooperative game, consequently if at least one allocation of this kind exists, the core is non-empty. The second meaning is strictly connected to a financial aspect (see \cite{denault.2001}):
an undercut happens when a portfolio allocation exceeds the risk capital that the whole group of institutions would face. \par
Relying on previous results, we can establish some sufficient conditions for positivity of $\Phi(c_W)$ and $\beta(c_W)$, i.e. to ensure that all their coordinates are non-negative, meaning that each distressed institution brings a positive marginal contribution to systemic risk.

\begin{proposition}\label{propertiesofShapleyandBanzhaf}
If for all $S \in 2^{\mathcal{D}}\setminus \emptyset$ and for all $j_k \in S$ the following hypotheses are verified:
\begin{enumerate}
\item $SCoVaR_{i\vert S\setminus \left\{j_k \right\}}^{\tau_1\vert \tau_2}>SCoVaR_{i\vert S}^{\tau_1\vert \tau_2}$;
\item $\sum_{j_k \in S\setminus \left\{j_k \right\}} X_{j_k}\leq ES_{\tau_2}\left(\sum_{j_k \in S\setminus \left\{j \right\}} X_{j_k}\right)$;
\item $\sum_{j_k \in S} X_{j_k}\leq ES_{\tau_2}\left(\sum_{j_k \in S} X_{j_k}\right)$,
\end{enumerate}
then $\phi_{j_k}(c_W) \geq 0$ and $\beta_{j_k}(c_W)$ for all $j_k \in S$.
\end{proposition}

\begin{proof}
Given a coalition of distressed institutions $S \neq \emptyset$ and any element $j_k \in S$, we can apply Proposition \ref{propertyofSCoES} to
two coalitions $S$ and $S \setminus \left\{j_k \right\}$ by reformulating its three hypotheses. Since by Proposition \ref{propertyofSCoES} we have that
$SCoES_{i\vert S\setminus \left\{j_k \right\}}^{\tau_1\vert \tau_2} \geq SCoES_{i\vert S}^{\tau_1\vert \tau_2}$,
then all terms in the sums in (\ref{ShapleycW}) and (\ref{BanzhafcW}) are positive.
\end{proof}
%
\section{Application}
\label{sec:application}
%
\noindent To illustrate how the $SCoES$ risk measure behaves in
practice we examine the evolution of European Sovereign Credit
Spreads (CDS) over a period that includes the Eurozone sovereign
debt crisis of 2012. Specifically, we investigate the evolution over
time of the Shapley--Value $SCoES$ induced by the cooperative Game
where the Germany acts as the only \qmo safe\qmcsp country, as
described in the previous sections. The potentially distressed
countries are: Belgium, France, Greece, Italy, Netherlands, Portugal
and Spain. We consider a panel of daily CDS spreads over the period
from July 21st, 2008 to December 30th, 2014 except for the Greece
for which the data are available only until March 8, 2012. We use
US--denominated sovereign CDS for each country using data
obtained from Datastream. Our aim is to assess how the events
related to the European sovereign debt crisis impact the safety of
the most important economy in the Euro area, using the provided risk
measure and the associated risk measurement framework based on the
cooperative game. A similar empirical investigation has been conducted by
\cite{bernardi_catania.2015} using stock market data of the major European financial indexes,
\cite{lucas_etal.2014} using dynamic Generalised Autoregressive Score (GAS) models on CDS, and \cite{engle_etal.2014} again using stock market data of European individual institutions, and \cite{blasques_etal.2014} using spatial GAS models on European sovereign debt CDS.\newline
\indent Major financial events affecting the Euro area
during the considered period are collected in Table
\ref{tab:fin_crisis_timeline}. Since EU countries have been affected by the crisis to different
degree, sovereign credit spreads in Europe are strongly correlated.
Figure \ref{fig:cds_log} shows the evolution of the credit default
spreads in log basis points for the period covered by our analysis.
Visual inspection of the series reveals clear common patterns
particularly between Netherlands and Germany on the one hand and
Italy and Spain on the other hand. As expected, the evolution of the
Greek CDS strongly differs from those of the other countries in the
sample. Summary statistics of log CDS returns multiplied by $100$
are reported in Table \ref{tab:cds_summary_stat}.\newline
\indent In order to calculate the $SCoES$ risk measure a parametric
assumption about the joint distribution of the involved CDS
log--returns should be made. Although not always supported by the
empirical evidence, we assume that the CDS returns are jointly
Gaussian. The Gaussian assumption is not only convenient but it represents also a
common choice for practical applications, and it favours the interpretation of the estimation
results as the output of a graphical model, see \cite{koller_friedman.2009}. Nevertheless the Gaussian distribution can be easily replaced by either another parametric distribution or by more involved dynamic models that describe the evolution over time of the CDS, see, for example, \cite{bernardi_catania.2015}. Proposition \ref{th:gauss_scovar_scoes} in Appendix
\ref{app:appendix_scoes_calc} provides the analytical formulas to
calculate $VaR$, $ES$, $SCoVaR$ and $SCoES$ under the Gaussian
assumption. As far as parameter estimation is concerned we apply the
Graphical--Lasso algorithm of \cite{friedman_etal.2008}, which allows for sparse covariance estimation. The tuning parameter that regulates the amount of sparsity in the covariance structure has been fixed at $\lambda_N=2\sqrt{\frac{\log p}{N}}$, where $N$ denotes the sample size, as suggested by the theory \cite{ravikumar_etal.2011}, see also \cite{hastie_etal.2015} and references therein.\newline
\indent To analyse more deeply the impact of the recent European
Sovereign debt crisis, we estimate recursively the SCoES over the
sample period using a rolling window. 
%
%
%
%
 %
%
\begin{table}[!t]
\captionsetup{font={small}, labelfont=sc}
\begin{center}
\begin{small}
\resizebox{0.9\columnwidth}{!}{%
\smallskip
\begin{tabular}{lcccccccccc}\\
\toprule
Name & Min & Max & Mean & Std. Dev. & Skewness & Kurtosis & 1\% Str. Lev. & JB \\
\hline
Belgium &  -21.912  & 13.854  & -0.049  & 2.726  & -0.422  & 12.144  & -7.999  & 5477.445  \\
France & -23.002  & 18.643  & 0.023  & 3.155  & -0.213  & 10.068  & -9.525  & 3257.265  \\
Germany & -33.747  & 30.839  & -0.017  & 3.367  & -0.368  & 21.499  & -9.535  & 22264.754  \\
Greece & -48.983  & 23.611  & 0.401  & 5.183  & -1.029  & 17.169  & -13.700  & 6081.685  \\
Italy & -42.675  & 34.358  & 0.011  & 4.237  & -0.579  & 17.933  & -12.229  & 14571.710  \\
Netherlands& -25.672  & 18.572  & -0.047  & 3.028  & -0.466  & 13.130  & -9.669  & 6721.984  \\
Portugal & -61.177  & 26.909  & 0.064  & 4.320  & -1.616  & 33.500  & -10.911  & 61104.578  \\
Spain & -35.180  & 27.174  & 0.020  & 4.137  & -0.078  & 11.616  & -11.781  & 4824.268  \\
\bottomrule
\end{tabular}}
\caption{Summary statistics of the panel of country specific CDS
spreads for the period beginning on July 21, 2008 and ending on
December 20, 2014. For the Greece the period goes from July 21, 2008
 to March 8, 2012. The seventh column, denoted by ``1\% Str.
Lev.'' is the 1\% empirical quantile of the returns distribution,
while the eight column, denoted by ``JB'' is the value of the
Jarque-Ber\'a test-statistics.}
\label{tab:cds_summary_stat}
\end{small}
\end{center}
\end{table}
%
\begin{table}[!h]
\centering
\resizebox{0.9\columnwidth}{!}{%
\begin{tabular}{ll}
\toprule
Date  & Event  \\
\hline
Mar. 9, 2009 & the peak of the onset of the recent GFC.\\
Oct. 18, 2009 & Greece announces doubling of budget deficit.\\
Mar. 3, 2010 & EU offers financial help to Greece.\\
Apr. 23, 2010 & Greek Prime Minister calls for Eurozone--IMF rescue package.\\
Apr. 23, 2010 & Greece achievement of 18bn USD bailout 
from EFSF, IMF and bilateral loans.\\
Nov. 29, 2010 & Ireland achievement of 113bn USD bailout 
from EU, IMF and EFSF.\\
May 05, 2011 & the ECB bails out Portugal.\\
July 21, 2011 & Greece is bailed out.\\
Dec. 22, 2011 & ECB launches the first Long-Term 
Refinancing Operation (LTRO).\\
Feb. 12, 2012 & Greece passes its most severe austerity package yet.\\
Mar. 1, 2012 & ECB launches the second LTRO.\\
\multirow{2}{*}{July 26, 2012} & unexpectedly, the ECB president 
Mario Draghi, announces that \\
&\qmo The ECB is ready to do whatever it takes to preserve the euro\qmc.\\
Oct. 8, 2012 & European Stability Mechanism (ESM) is inaugurated.\\
\multirow{2}{*}{April 07, 2013} & the conference of the Portuguese 
Prime Minister regarding \\
& the high court's block of austerity plans.\\
Aug. 23, 2013 & the Eurozone crisis leads to more bankruptcies in Italy.\\
Sep. 12, 2013 & European Parliament approves new unified bank supervision system.\\
\bottomrule
\end{tabular}}
\caption{Financial crisis timeline.}
\label{tab:fin_crisis_timeline}
\end{table}
%

\begin{figure}[!ht]
\begin{center}
\captionsetup{font={small}, labelfont=sc}
\includegraphics[width=0.7\linewidth]{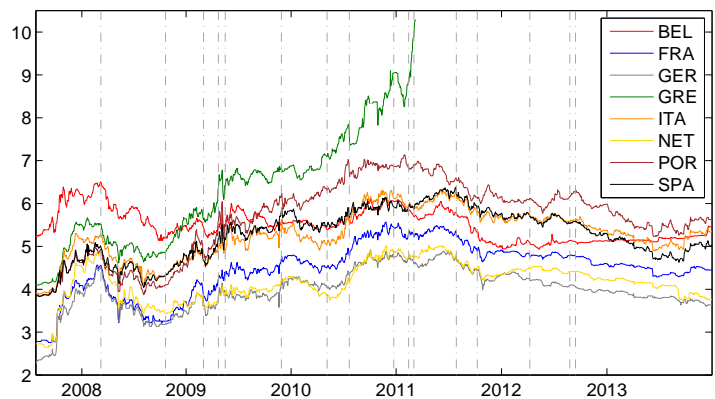}
\caption{Daily evolution of Credit default swap spreads of eight
European sovereigns, from July 21, 2008 to December 20, 2014 in log
basis points. Vertical dashed lines represent major financial
downturns: for a detailed description see Table
\ref{tab:fin_crisis_timeline}.}
\label{fig:cds_log}
\end{center}
\end{figure}
%
\begin{figure}[!t]
\begin{center}
\captionsetup{font={small}, labelfont=sc}
\subfloat[Prior to the Greek crisis]{\label{fig:sv_scoes_with_greece}\includegraphics[width=0.45\textwidth]{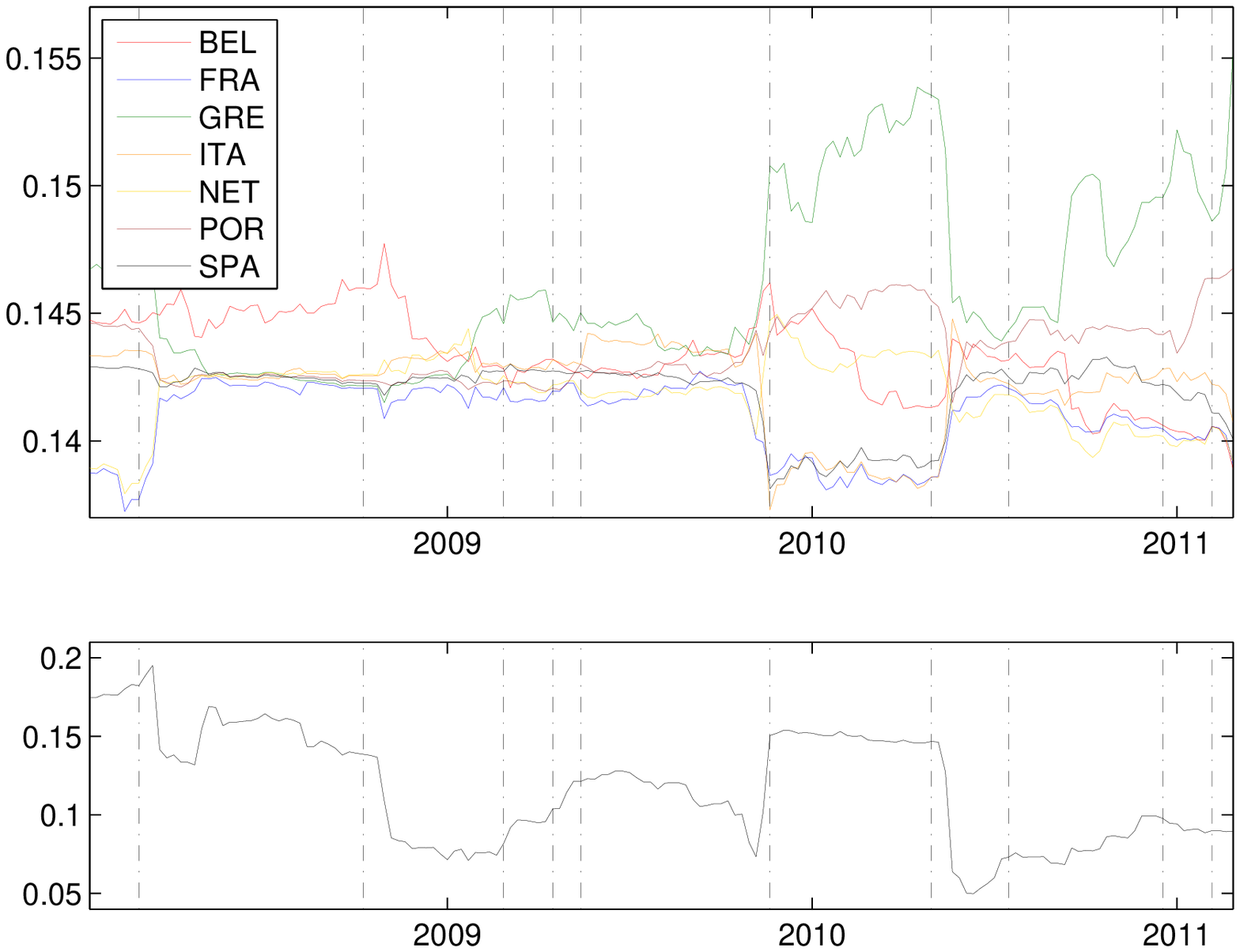}}
\qquad
\subfloat[After the Greek crisis]{\label{fig:sv_scoes_without_greece}\includegraphics[width=0.45\textwidth]{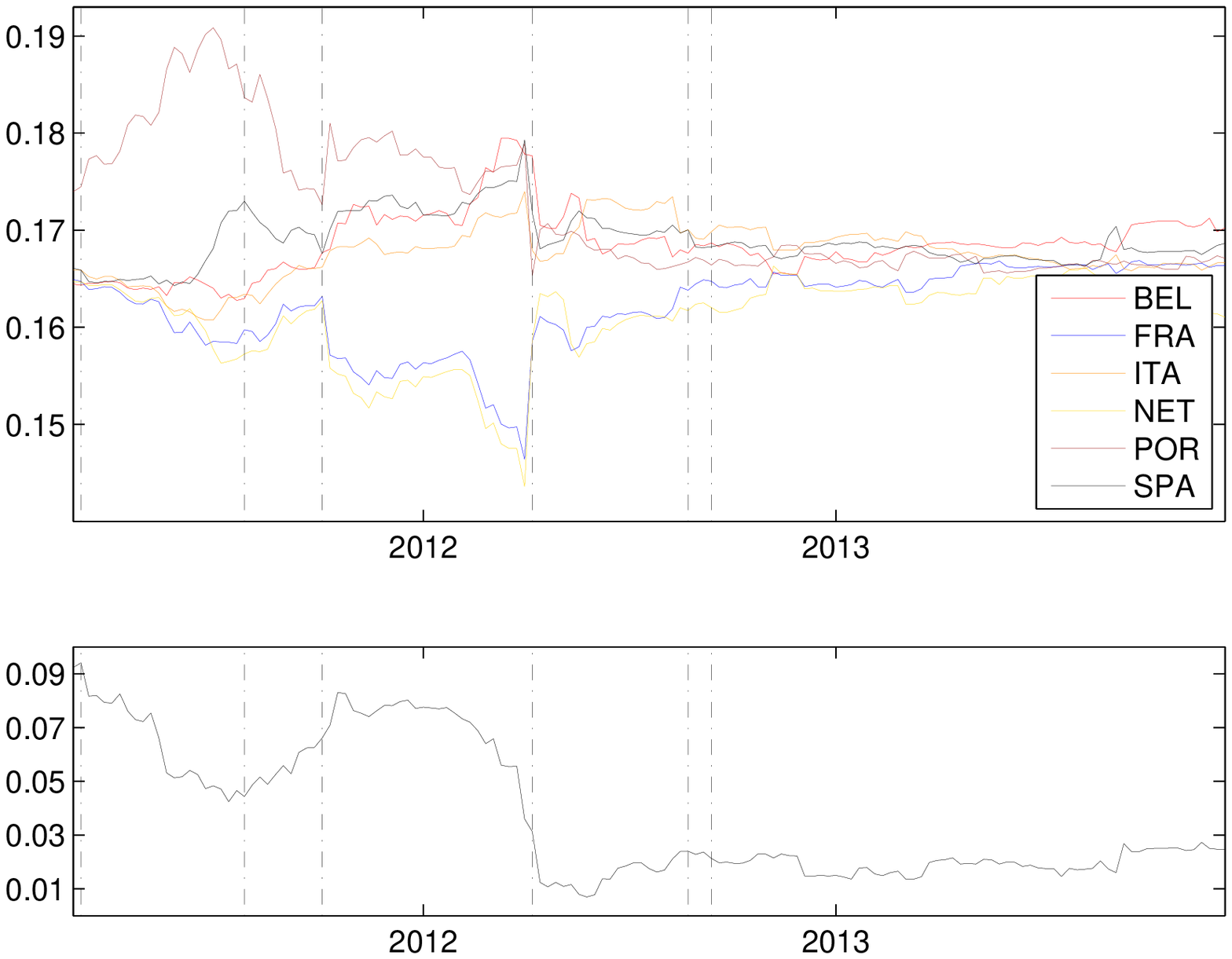}}
\caption{\textit{(Top panel):} Shapley Value using the $SCoES$ at
the $5\%$ level of Germany with respect to all the remaining
countries. \textit{(Bottom panel):} Overall risk of Germany measured
by the $SCoES$ when all the remaining countries are in distress.
Vertical dashed lines represent major financial downturns: for a detailed description see Table \ref{tab:fin_crisis_timeline}.}
\label{fig:SV_SCoES_results}
\end{center}
\end{figure}
%
\begin{figure}[!t]
\begin{center}
\captionsetup{font={small}, labelfont=sc}
\subfloat[Prior to the Greek crisis]{\label{fig:bz_scoes_with_greece}\includegraphics[width=0.45\textwidth]{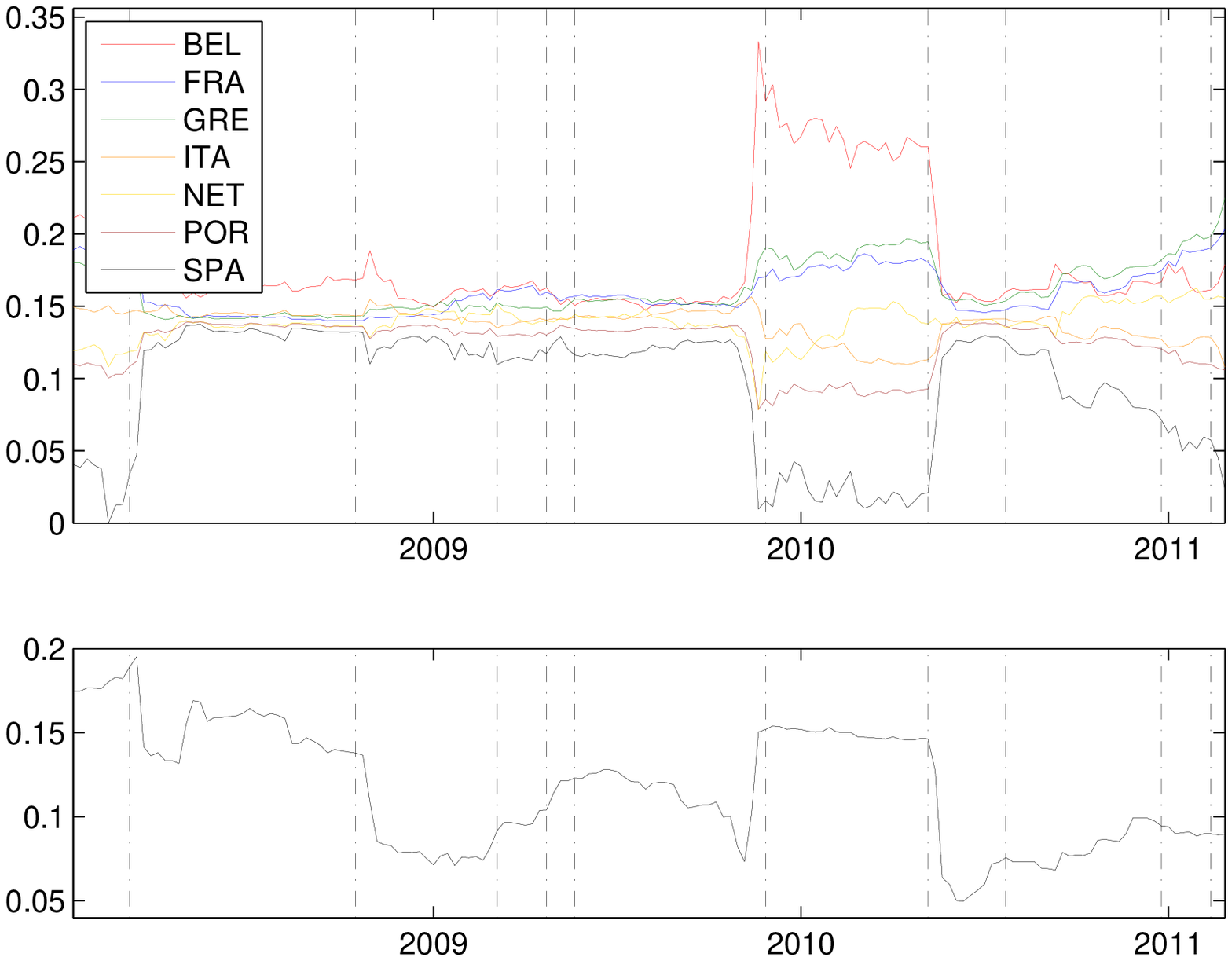}}
\qquad
\subfloat[After the Greek crisis]{\label{fig:bz_scoes_without_greece}\includegraphics[width=0.45\textwidth]{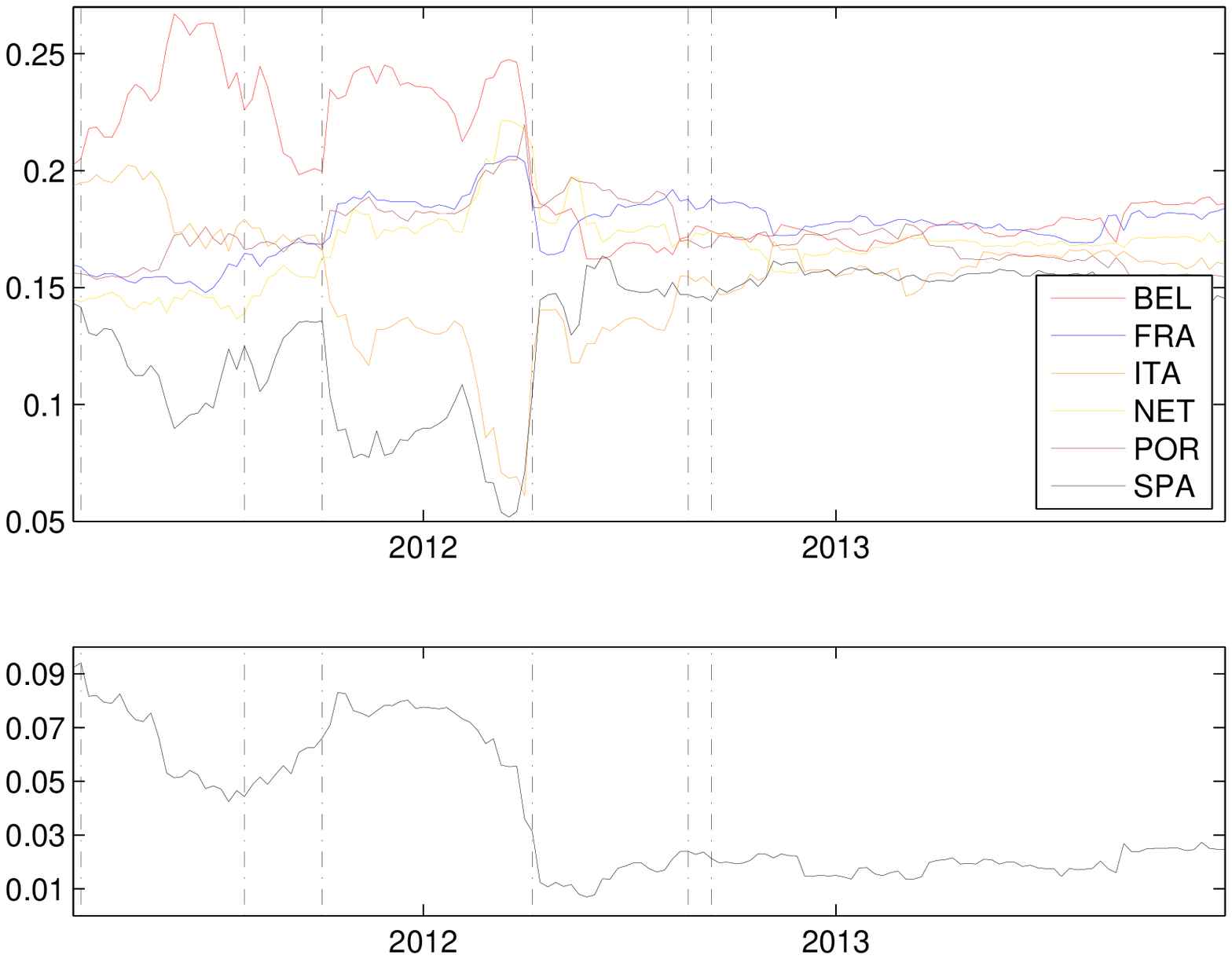}}
\caption{\textit{(Top panel):} Banzhaf Value using the $SCoES$ at
the $5\%$ level of Germany with respect to all the remaining
countries. \textit{(Bottom panel):} Overall risk of Germany measured
by the $SCoES$ when all the remaining countries are in distress.
Vertical dashed lines represent major financial downturns: for a detailed description see Table \ref{tab:fin_crisis_timeline}.}
\label{fig:BZ_SCoES_results}
\end{center}
\end{figure}
%
Moreover, to obtain results more robust to temporary
short--term shocks affecting the considered economies, we consider
weekly log--CDS returns. Specifically, at each point in time we
estimate the $SCoES$ risk measure using a window of $N=26$ more recent
weekly observations, and then we run the cooperative game to get the
Shapley Value with $\tau_1=\tau_2=5\%$. It is worth mentioning that, with $p=8$ institutions, we are going to estimate $d=44\gg 26=N$ parameters. The Graphical--Lasso method of \cite{friedman_etal.2008} delivers consistent estimates of the parameters even when the number of parameters is greater than the dimension of the sample, see \cite{hastie_etal.2015} and \cite{tibshirani.1996,tibshirani.2011} for further details.\newline
\indent Results are reported in Figure
\ref{fig:SV_SCoES_results}, for the period before the onset of the
Greek crisis \eqref{fig:sv_scoes_with_greece}, as well as for the
subsequent period \eqref{fig:sv_scoes_without_greece}. For both
periods, the bottom panel reports the overall impact of the distress
of the remaining countries over the German economy as measured by
the $SCoES$. Throughout the sample period, the overall risk of
Germany due to the potential distress of one or more of the
remaining European countries is high, until mid 2011. At that
time, the level suddenly decreases to the lower level of about 0.09,
as a consequence of the bailouts of Portugal and Greece, in May and July 2011, respectively. The overall risk still remain at the level of 0.09 till the April 2013 a few months later the announcement of the implementation of the of the Outright Monetary
Transactions (OMT) and the European Stability Mechanism (ESM) in October 2012. If is worth noting that the launch of the first Long--Term Refinancing Operation (LTRO) by ECB in December 2011 and the second LTRO in March 2012 only had a moderate impact on the overall risk that decreased till mid 2012 and the unexpected strongest defence of the Euro of the ECB President Mario Draghi (July 26, 2012), did not contribute to reduce the risk of the Germany.\newline
%
%
%
\indent Concerning the evolution of the Shapley values reported in the top
panel of Figure \ref{fig:SV_SCoES_results} and the Banzhaf Value in Figure \ref{fig:BZ_SCoES_results}, they can be interpreted
as the normalised country risk factors. As expected, the two approaches provide different contributions to risk, in particular during the onset of the European crisis, reflecting their different properties. More precisely, the Shapley solution, suggests that, during the
period in between the bailout of Ireland (November 29, 2010) and the
bailout of Portugal (May 5, 2011), the most severe source of risk 
for the Germany's economy is represented by the Greece, while Spain,
Italy and France contribute less, see Figure \ref{fig:SV_SCoES_results}. The picture by Banzhaf is a little bit different since, during the same period, the most important source of risk for Germany is Belgium, followed by Greece, a result which is a little bit surprisingly. Afterwards, the overall
contributions of European countries converge, for both methods. The failure of the
Greek austerity package in February 2012 suddenly increases the
riskiness of Greece which is comparable only with that of Portugal
at the beginning of 2011. Interestingly, the proposed approach is able to capture the most important events that happened during European sovereign debt crisis of 2012, as reported in Table \ref{tab:fin_crisis_timeline}.

%
\section{Conclusions and further developments}
\label{sec:conclu}
%
This paper presents a cooperative game among distressed institutions
to assess the potential damage done by all possible coalitions in
distress. At this aim, a new risk measure which features some
properties of the standard Expected Shortfall in a financial
framework where some institutions are distressed and contagion
threatens the remaining safe institutions is developed. Standard
solution concepts like Shapley value and Banzhaf value can be
helpful to measure the marginal contributions to systemic risk.\newline
\indent Our study of the European sovereign debt crisis of 2012 provides empirical support for
the ability of the proposed cooperative Game approach to systemic
risk measurement to effectively capture the dynamic evolution of the
overall riskiness of the European countries. Furthermore, the proposed risk measurement
framework is able to identify the major sources of risk and the risk
contributions.\newline
\indent Further extensions of such a theoretical setup can be
conceived, in terms of more complex and precise cost functions to be
employed in the cooperative game. Moreover, a detailed analysis of
correlations among distressed institutions might give rise to
different game structures, such as a priori unions or bounded forms
of cooperation which can be described with the help of graphs. In such cases, a
model with some constraints might be necessary to determine the
characteristics of risk transmission and the related consequences on
the systemic risk. It is possible that a given structure with
certain links among institutions can minimize the contagion risk.
%
\section*{Acknowledgements}
This research is supported by the Italian Ministry of Research PRIN
2013--2015, ``Multivariate Statistical Methods for Risk Assessment''
(MISURA), and by the ``Carlo Giannini Research Fellowship'', the
``Centro Interuniversitario di Econometria'' (CIdE) and ``UniCredit
Foundation''. We would like to thank Rosella Castellano, Umberto
Cherubini, Rita D'Ecclesia, Fabrizio Durante, Piotr Jaworski, Viviana Fanelli,
Gianfranco Gambarelli, Sabrina Mulinacci, Roland Seydel, Marco
Teodori, the audience at CFE-ERCIM 2014 in Pisa (Italy), and then
audience at EURO 2016 in Poznan (Poland) for their valuable comments
and suggestions. The usual disclaimer applies.
%
\appendix
\section*{Appendix}
\label{app:appendix_scoes_calc}
%
\noindent In this Appendix we provide analytical formulas for the computation of the $VaR$, $ES$, $SCoVaR$ and $SCoES$, as formally defined in the previous sections, under the assumption of the joint Gaussian distribution of the involved variables.

\begin{proposition}
\label{th:gauss_scovar_scoes}
Let $X \sim \mathcal{N}_p \left(\mu,\Sigma\right)$ where $\mu\in\mathbb{R}^p$ is a vector of location parameters and $\Sigma$ is a $\left(p\times p\right)$-- symmetric variance-covariance matrix. Consider the transformation $Z=\left[X_i,\sum_{k=1,k\neq i}^p X_{k}\right]^\prime$, for $i=1,2,\dots,p$, then $Z\sim\mathcal{N}_2\left(\mu_Z,\Sigma_Z\right)$, with
\begin{align}
\mu_Z&=\left[\mu_i,\ \sum_{k=1,k\neq i}^p\mu_k\right]^\prime
\label{muZ}\\
\Sigma_Z&=\left(\begin{matrix}\sigma^2 (X_i) & \sum_{k=1,k\neq i}^p \sigma_{i,k} \\ & & \\ \sum_{k=1,k\neq i}^p \sigma_{i,k} &
\sigma^2 \left( \sum_{k=1,k\neq i}^p X_{k} \right)
\end{matrix}\right),
\label{SigmaZ}
\end{align}
where $\sum_{k=1,k\neq i}^p \sigma_{i,k}$ denotes the covariance between $X_i$ and $\sum_{k=1,k\neq i}^p X_{k}$, and $\sigma^2(\cdot)$ denotes the variance of $X_i$. Under the previous assumptions the VaR, ES, of $X_i$, for $i=1,2,\dots,p$ are calculated as follows:
\begin{align}
\widehat{\nu}_{i}^{\tau_2}&\equiv VaR_{\tau_2}\left(\sum_{k=1,k\neq i}^p X_{k} \right)\nonumber\\
&=\sum_{k=1,k\neq i}^p\mu_k+\sqrt{\sigma^2\left(\sum_{k=1,k\neq i}^pX_k\right)}\Phi^{-1}\left(\tau_2\right)\\
\widehat{\psi}_{i}^{\tau_2}&\equiv ES_{\tau_2}\left(\sum_{k=1,k\neq i}^p\mu_k\right)\nonumber\\
&=\sum_{k=1,k\neq i}^p\mu_k-\sqrt{\sigma^2\left(\sum_{k=1,k\neq i}^p X_k\right)}\dfrac{\phi\left(\widehat{\nu}_{i}^{\tau_2}\right)}{\Phi\left(\widehat{\nu}_{i}^{\tau_2}\right)},
\end{align}
see, \cite{nadarajah_etal.2014} and \cite{bernardi.2013}, while the SCoVaR and SCoES becomes
\begin{align}
\hat{\gamma}_{i\vert \sum_{k\neq i}X_k}^{\tau_1\vert\tau_2}&\equiv SCoVaR_{i\vert \sum_{k\neq i}X_k}^{\tau_1\vert \tau_2}\nonumber\\
&:=y\in\mathbb{R}\ni\frac{F_{X_i,\sum_{k\neq i} X_{k}}\left(y,\sum_{k=1,k\neq i}^p X_{k}\leq
\widehat{\nu}_{i}^{\tau_2}\right)}{\tau_2}=\tau_1\label{eq:scovar_eq}\\
\hat{\varsigma}_{i\vert\sum_{k=1,k\neq i}X_k}^{\tau_1\vert \tau_2}&\equiv SCoES_{i\vert\sum_{k\neq i}}^{\tau_1\vert \tau_2}\nonumber\\
&=\sum_{k\neq i}^p\mu_k-\sqrt{\sigma^2\left(X_i\right)}\left[\phi\left(\hat{\gamma}_{i\vert\sum_{k\neq i}X_k}^{\tau_1\vert\tau_2}\right)\Phi\left(\frac{\widehat{\nu}_{i}^{\tau_2}-\rho\hat{\gamma}_{i\vert\sum_{k\neq i}X_k}^{\tau_1\vert\tau_2}}{\sqrt{1-\rho^2}}\right)\right.\nonumber\\
&\qquad\qquad\qquad\qquad\qquad\qquad\left.+\rho\phi\left(\widehat{\nu}_{i}^{\tau_2}\right)
\Phi\left(\frac{\hat{\gamma}_{i\vert\sum_{k\neq i}X_k}^{\tau_1\vert\tau_2}-\rho\widehat{\nu}_{i}^{\tau_2}}{\sqrt{1-\rho^2}}\right)\right],
\end{align}
for $i=1,2,\dots,p$, where $F_{X,Y}\left(\cdot\right)$ denotes the joint cdf of the random variables $\left(X,Y\right)$. 
\end{proposition}
\noindent Equation \eqref{eq:scovar_eq} implicitly defines the $SCoVaR$ as the value of $y$ that solves the conditional cdf of the involved variables equal to $\tau_1$. The solution always exists and is unique because the involved random variables are absolutely continuous.
%
\newpage
\clearpage

\end{document}